\documentclass[letterpaper, 10 pt, conference]{IEEEtran}

\usepackage{amsmath,amsthm}
\usepackage{amssymb}
\usepackage[ruled,vlined]{algorithm2e}
\usepackage{algpseudocode}
\usepackage{mathrsfs}
\usepackage{tikz}
\usepackage{multirow}
\usepackage{url}
\usepackage{flushend}
\usepackage{comment}

\newtheorem{definition}{Definition}
\newtheorem{example}{Example}
\newtheorem{proposition}{Proposition}
\newtheorem{lemma}{Lemma}
\newtheorem{theorem}{Theorem}
\newtheorem{corollary}{Corollary}

\title{\LARGE \bf
Optimal Representation for Right-to-Left \\ 
Parallel Scalar Point Multiplication
}

\author{
\IEEEauthorblockN{Kittiphon Phalakarn}
\IEEEauthorblockA{Department of Computer Engineering\\
Chulalongkorn University\\
email: kittiphon.p@student.chula.ac.th\vspace{-0.4cm}}
\and
\IEEEauthorblockN{Kittiphop Phalakarn}
\IEEEauthorblockA{Department of Computer Engineering\\
Chulalongkorn University\\
email: kittiphop.ph@student.chula.ac.th \vspace{-0.4cm}}
\and
\IEEEauthorblockN{Vorapong Suppakitpaisarn}
\IEEEauthorblockA{Department of Computer Science\\
The University of Tokyo \\
email: vorapong@is.s.u-tokyo.ac.jp\vspace{-0.4cm}}}

\begin{document}

\maketitle 
\thispagestyle{empty}
\pagestyle{empty}

\begin{abstract}
This paper introduces an optimal representation for a right-to-left parallel elliptic curve scalar point multiplication. The right-to-left approach is easier to parallelize than the conventional left-to-right approach. However, unlike the left-to-right approach, there is still no work considering number representations for the right-to-left parallel calculation. By simplifying the implementation by Robert, we devise a mathematical model to capture the computation time of the calculation. Then, for any arbitrary amount of doubling time and addition time, we propose algorithms to generate representations which minimize the time in that model. As a result, we can show a negative result that a conventional representation like NAF is almost optimal. The parallel computation time obtained from any representation cannot be better than NAF by more than 1\%.
\end{abstract}

\renewcommand\IEEEkeywordsname{Keywords}

\section{Introduction}\label{introduction}

Scalar point multiplication $nP$ 
is an 
important operation in elliptic curve cryptography. Many techniques were proposed to improve this operation. Since binary representation of $n$ affects the speed of ``double-and-add" scalar point multiplication, there are techniques, including sliding window \cite{slidingwindow1,slidingwindow2}, non-adjacent form (NAF) \cite{naf}, window NAF ($w$NAF) \cite{wnaf}, and fractional $w$NAF \linebreak (f-$w$NAF) \cite{fwnaf}, trying to improve the binary representations. 

Elliptic curve scalar point multiplication is a special case of exponentiation 
$g^n$ for 
a group member $g$. 
When 
$g^{-1}$ can be calculated easily, many algorithms for double-and-add, including 
the proposed one, can be applied to reduce the computation of the square-and-multiply for the exponentiation. Although all notations in this paper are for scalar point multiplication, our contribution also includes those general arithmetic operations.

To improve scalar point multiplication further, parallelism is used. 
Many schemes using parallelism \cite{garcia2002parallel}, \cite{nocker}, \cite{borges2017parallel} were proposed.
However, these algorithms can be significantly improved if the number of processors we have is two (see proof in Section~\ref{nocker}). A parallel algorithm 
is also proposed 
in~\cite{izu2002fast}. However, the authors aim to find an algorithm that can resist side channel attacks. 
By that, the algorithm requires more steps and is slightly slower than those that do not resist the attacks.

Most of the work 
are based on left-to-right approach, because, 
the technique can be sped up using precomputation points~\cite{wrNAF}. However, when 
parallelism is used, right-to-left technique is known to be easier \cite{inscrypt}. There are not many works using 
right-to-left technique. The only technique we know is 
in~\cite{munro}, which we call as ``parallel double-and-add". Although the technique is optimal, 
it assumes that point doubling and addition use same amount of time. This is 
not true for scalar point multiplication, 
e.g. in the fastest 
scalar point multiplication, twisted Edwards curve, addition takes 50\% more time than doubling \cite{twistedEd}. In other curves, 
addition takes significantly more time than doubling \cite{web}. Thus, this algorithm is not optimal for scalar point multiplication.

That motivates us to consider an optimal representation for parallel right-to-left scalar point multiplication. 
Based on right-to-left method by Moreno and Hasan in~\cite{moreno2011spa} and its implementation by Robert in~\cite{inscrypt}, 
we propose a new time model and problem for parallel scalar point multiplication with arbitrary amount of doubling time $D$ and addition time $A$. Then, we show a negative result that 
NAF is almost optimal for any arbitrary amount of doubling time and addition time.

To show that NAF is almost optimal, we propose algorithms to generate optimal representations for all cases. 
Then, we prove and perform numerical experiment to show that the average 
time obtained from NAF is very close to the optimal. 
The difference between the average 
time obtained from the algorithms is not more than $1\%$ in all experimental settings.

Although techniques such as $w$NAF or f-$w$NAF provide much faster scalar point multiplications than NAF in left-to-right setting~\cite{nocker}, they 
are not better in our parallel setting. Our results indicate that there is no representation can improve the average computation time of NAF by more than $1\%$, while, in those schemes, we have to perform some precomputation tasks before the calculation for scalar point multiplication.


\section{Preliminaries}\label{preliminaries}

\subsection{Binary Representation}\label{binary}

We first introduce the notation of binary representation used in this paper. We assume that $n$ is an input and all representations are of $n$ if not state otherwise.

\begin{definition}[Binary representation of $n$ using digit set $\mathcal{S}$]
Let $n,\lambda \in \mathbb{Z}_{\ge 0}$ and $\{0,1\} \subseteq \mathcal{S} \subseteq \mathbb{Z}$. The set of all possible binary representations of $n$ using digit set $\mathcal{S}$ is $\mathscr{N}_\mathcal{S} = \{N_\mathcal{S} = n_{\lambda}...n_0\}$ such that $\displaystyle \sum_{i=0}^{\lambda}n_i2^i=n$, $n_i \in \mathcal{S}$ for all $0 \le i \le \lambda$. We use $\bar{s}$ instead of $-s$ to simplify the notation. If $\mathcal{S}=\mathcal{B}=\{0,1\}$, we call $\mathscr{N}_\mathcal{B}=\{N_\mathcal{B}\}$ set of all binary representations of $n$. If $\mathcal{S}=\mathcal{C}=\{\bar{1},0,1\}$, we call $\mathscr{N}_\mathcal{C}=\{N_\mathcal{C}\}$ set of all canonical binary representations of $n$.
\end{definition}

Note that $N_\mathcal{B}$ is unique while $N_\mathcal{C}$ is not. In this paper, we use regular expression to represent 
patterns of digits, e.g. $10^3\bar{1}$ means $1000\bar{1}$, and $1^*$ means zero or more `1's.

\subsection{Non-Adjacent Form (NAF)}\label{naf}

NAF \cite{naf} is one technique used to improve scalar point multiplication. It changes $N_\mathcal{B}$ to $N_\mathcal{C}$ with no consecutive non-zero digits which is proved to have minimal Hamming weight \cite{naf}. The algorithm can be described as in Algorithm 1. Note that NAF representation is unique, and we have a starting index $i$ as input for further use in our proposed algorithms.

\begin{algorithm}[ht]
\DontPrintSemicolon
\caption{Transform binary representation to non-adjacent form (toNAF)}
\SetKwInOut{Input}{input}
\SetKwInOut{Output}{output}
\Input{$N_\mathcal{B} = n_\lambda...n_0$, starting index $i$}
\Output{$N'_\mathcal{C} = n'_{\lambda+1}...n'_0$ with no consecutive non-zero digits in $n'_{\lambda+1}...n'_i$}
\Begin{
	$N'_\mathcal{C} \leftarrow N_\mathcal{B}$\;
	$n'_{\lambda+1} \leftarrow 0$\;
	\While {$i<\lambda$} {
        \If {$n'_i=1$ \textnormal{\textbf{and}} $n'_{i+1}=1$} {
			$n'_i \leftarrow \bar{1}$\;
            $i \leftarrow i+1$\;
            \While {$n'_i=1$} {
            	$n'_i \leftarrow 0$\;
                $i \leftarrow i+1$\;
			}
            $n'_i \leftarrow 1$\;
        }
        \lElse {$i \leftarrow i+1$}
    }
    \Return $N'_\mathcal{C}$
}
\end{algorithm}

In short, Algorithm 1 changes $01^p$ to $10^{p-1}\bar{1}$ for $p \ge 2$ from $n'_i$ to $n'_{\lambda+1}$. To get NAF representation, we set $i=0$.

\begin{example}
Consider $n=371$ with $N_\mathcal{B}=101110011$. If $i=0$, we have $N'_\mathcal{C}=10\bar{1}00\bar{1}010\bar{1}$, and if $i=5$, we have $N'_\mathcal{C}=10\bar{1}0\bar{1}10011$. \hfill $\square$
\end{example}

\subsection{``Parallel Double-and-Add" Scalar Point Multiplication}\label{double_and_add}

Borodin and Munro \cite{munro} presented Theorem 1 with a technique we can apply to ``parallel double-and-add" technique. This technique does scalar point multiplication in $\lceil\log_2 n\rceil$ steps. Notice that the time is measured in ``steps" which means that point doubling and addition use same amount of time.

\begin{theorem}[adapt from Lemma 6.1.1 in \cite{munro}]
Let $n \in \mathbb{Z}_+$ and $P$ be elliptic point. $nP$ can be computed from $P$ using two processors in $\lceil \log_2 n \rceil$ steps.
\end{theorem}

``Parallel double-and-add" technique uses two processors: doubling processor and addition processor. Doubling processor calculates $2P, 4P, ..., 2^iP$, and addition processor adds the result from doubling processor cumulatively according to $N_\mathcal{B}=n_\lambda...n_0$. 
Note that for the least significant `1', e.g. $n_0$ in Example 2, addition processor can ``copy" the result from doubling processor with no addition. 

\begin{example}
To calculate $87P$ using $N_\mathcal{B}=1010111$, we need 7 steps as depicts in Fig. 1.
\begin{figure}[ht]
    \centering
    \includegraphics[width=0.8\linewidth]{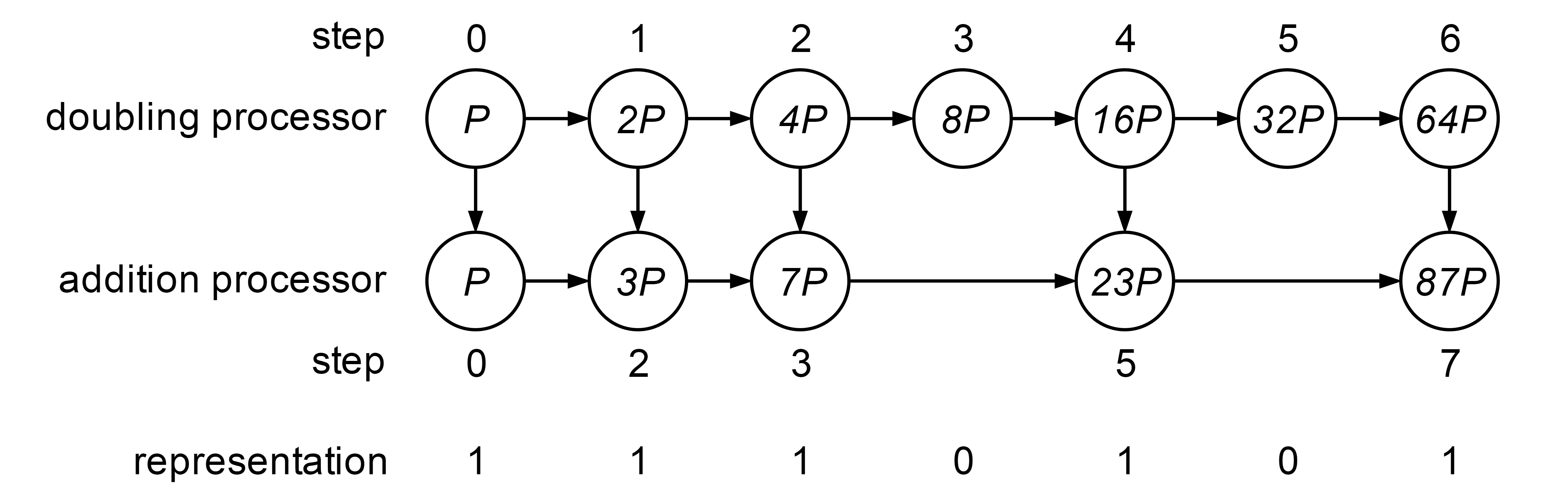}
    \caption{``Parallel double-and-add" scalar point multiplication for $N_\mathcal{B}=1010111$}
\end{figure}
\hfill $\square$
\end{example}

It is known that, in every elliptic curves, point addition takes more computation time than doubling. Hence, this applied model may not be practical for scalar point multiplication. 

\section{Our Calculation Model}\label{our_model}

We improve the model 
by defining the time used in ``parallel double-and-add" scalar point multiplication as follows.

Let $D\in \mathbb{R}_{\ge 0}$ be the amount of time doubling processor uses for one doubling, and $A\in \mathbb{R}_{\ge 0}$ be the amount of time addition processor uses for one addition. To calculate $nP$ using $N_\mathcal{S}$, addition processor considers each digit from $n_0$ to $n_\lambda$. Before we encounter the least significant non-zero digit, the time used is 0. At the least significant non-zero digit $n_i$, we wait for $2^iP$ to be finished at time $iD$, copy $2^iP$ or $-2^iP$ to addition processor using negligible additional time, and add/subtract $2^iP$ to/from addition processor $|n_i|-1$ times, so the computation time is now $iD+(|n_i|-1)A$. For later zero digits, the computation time is unchanged. And, for other non-zero digits $n_j$, we need to wait until doubling processor finishes $2^jP$ at time $jD$ and until addition processor finishes its previous work, then we add/subtract $2^jP$ to/from the current result $|n_j|$ times which uses $|n_j|A$ time units. 

\begin{example}
To calculate $87P$ using $\mathcal{S}=\{\bar{3},\bar{2},\bar{1},0,1,2,3\}$ and $N_\mathcal{S}=220\bar{1}\bar{3}1$ with $D=2$ and $A=3$, we need 26 time units as shown in Fig. 2. \hfill $\square$
\end{example}

\begin{figure}[ht]
    \centering
    \includegraphics[width=0.8\linewidth]{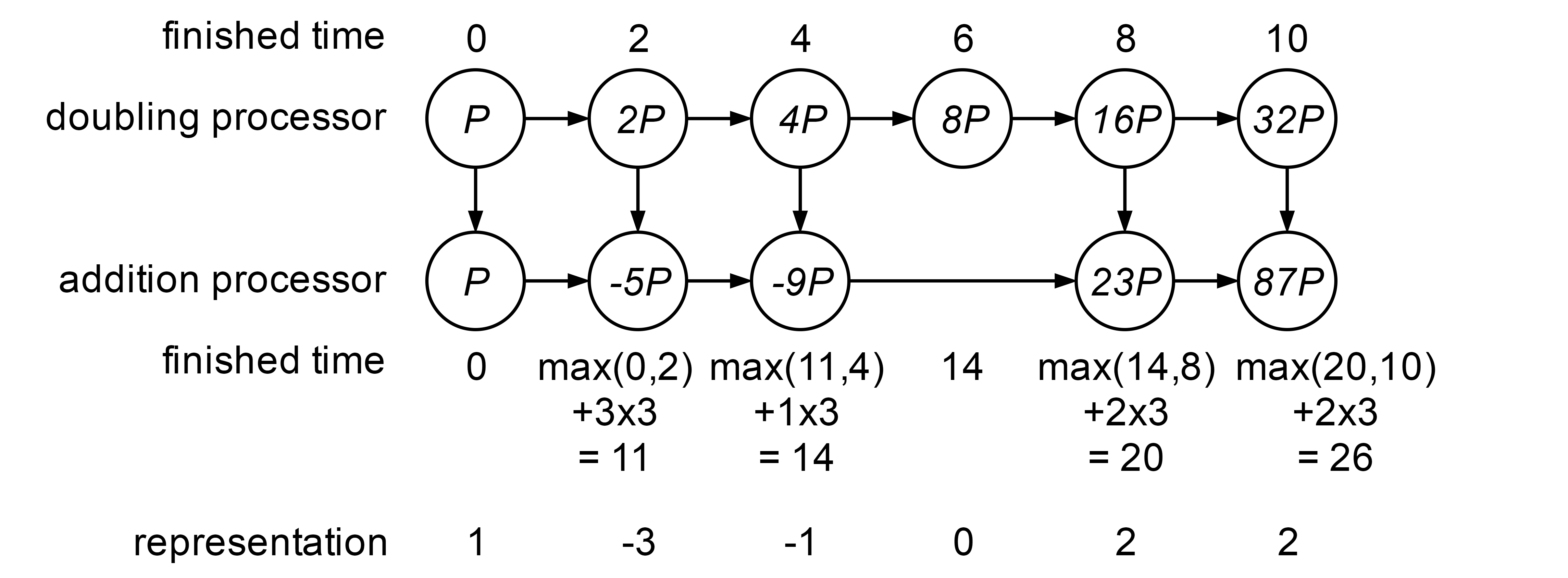}
    \caption{``Parallel double-and-add" scalar point multiplication for $D=2$, $A=3$, and $N_\mathcal{S}=220\bar{1}\bar{3}1$}
\end{figure}

Hence, we define the time model as follows:

\begin{definition}[Computation time of parallel scalar point multiplication]
Let $D\in \mathbb{R}_{\ge 0}$ be the time processor used for one doubling, $A\in \mathbb{R}_{\ge 0}$ be the time used for one addition, and $n \in \mathbb{Z}_+$ with binary representation $N_\mathcal{S}=n_{\lambda}...n_0$ using digit set $\mathcal{S}$. The computation time of parallel scalar point multiplication $nP$ using $N_\mathcal{S}$ after calculating from $n_0$ up to $n_i$ is $T(N_\mathcal{S},i)$ which can be calculated by:\\\\
$T(N_\mathcal{S},i):=$
$$
\begin{cases}
0 & \textrm{if } i=0 \textrm{ and } n_i=0;\\
T(N_\mathcal{S},i-1) & \textrm{if } i>0 \textrm{ and } n_i=0;\\
iD+(|n_i|-1)A & \textrm{if } i\ge 0 \textrm{, } n_i \neq 0 \textrm{, and }\\
& n_j=0; \forall j, 0 \le j < i;\\
\max(T(N_\mathcal{S},i-1),iD)+|n_i|A & \textrm{otherwise.}
\end{cases}
$$
\end{definition}

From Definition 2, the time used to calculate $nP$ using $N_\mathcal{S}$ is $T(N_\mathcal{S},\lambda)$. Our problem is to find 
$N_\mathcal{S}$ which uses minimum 
$T(N_\mathcal{S},\lambda)$ for given 
$n,D$ and $A$. 
Here, we consider only 
$0<D \le A$ since when $D=0$, minimal total Hamming weight representation 
\cite{mthw} is optimal; and $D \le A$ for all elliptic curve implementations that have been proposed up to this state~\cite{web}.

We do not consider communication time, denoted as $S$, between processors in our model, and will consider them as future work. Because $S$ 
is usually large, 
one might think 
we should add that time into the time of double, i.e. 
$D' = D + S$. 
However, doing that is 
too pessimistic. By pipelining, we need 
$D + D + S = 2D + S$ 
to do 
two doubling, 
not $2D' = 2D + 2S$. It is straightforward to show that the time until the point $2^iP$ is no more than $i(D + S)$. Hence, by assigning $D$ to $D + S$, one can calculate an upper bound of the computation time in our model. We strongly believe that the optimal representation will not significantly change by adding $S$. 


Moreover, we can see that $T(N_\mathcal{S},i)$ is always in the form $pA+qD$ for some integers $p$ and $q$, so from now on we will normalize $A$ and $D$ by assuming that $D=1$.

\begin{example}
Let 
$D=1$, $A=3$, and $n = 29$. If we calculate $29P$ using $N_\mathcal{B} = 11101$, we have $T(N_\mathcal{B},\lambda)=11$. If we use $N_\mathcal{C}=1000\bar{1}\bar{1}$, we have $T(N_\mathcal{C},\lambda)=8$. In Section~\ref{proof_algo1}, we will show that this is optimal among all representations. 
\hfill $\square$
\end{example}

\section{Properties of Our Model}\label{properties}


\subsection{Comparing computation time with N\"{o}cker's Algorithm}\label{nocker}

N\"{o}cker \cite{nocker} proposed an algorithm to distribute workload to compute $nP$ among all processors. The computation time used by his algorithm is as follows.
\begin{theorem}[adapt from Theorem 1 in \cite{nocker}]
Let $p\ge 2$ be the number of processors used to compute $nP$, $N_\mathcal{B}=n_\lambda...n_0$, and $c=D / A$. For $0<D \le A$, the computation time used by N\"{o}cker's algorithm is
$$
T \le \lambda D+\left(\frac{c}{{(1+c)}^p-1}(\lambda+1)-1+\lceil\log_2 p\rceil\right)A.
$$
Furthermore, this bound is tight since there is an integer $n$ which achieves this computation time.
\end{theorem}

\noindent When $p=2$ and $D=1$, we have $T \le \lambda+\left(\frac{A}{2A+1}(\lambda+1)\right)A.$ We will show in Corollary 1 that, when $0<2D\le A$, the computation time from our model is no more than
$\frac{1}{2}\lambda A+A+1,$
and, in Corollary 2, we show that, when $0<D\le A<2D$, the computation time from our model is no more than $A+\lambda+1$. Using these facts, we will show in the next theorem that our representations give better computation time.

\begin{theorem}
Algorithm 2 and 3 generate representations which have better worst case computation time than N\"{o}cker's algorithm when $A\ge 1$ and $\lambda \ge \max\left(\frac{4}{3}A+\frac{4}{3},3+\frac{3}{A}\right).$
\end{theorem}

\begin{proof}
Since $A\ge 1$ and $\lambda \ge \frac{4}{3}A+\frac{4}{3}$, we have
\begin{align*}
&\lambda + \left(\frac{A}{2A+1}(\lambda+1)\right)A \ge \lambda+\left(\frac{1}{2}-\frac{1}{4A+2}\right)\lambda A \\
&\ge \frac{1}{2}\lambda A + \lambda - \left(\frac{A}{4A+2}\right) \lambda
                                                 \ge \frac{1}{2}\lambda A + \frac{3}{4}\lambda 
                                                \ge \frac{1}{2}\lambda A + A + 1,
\end{align*}
and since $A\ge 1$ and $ \lambda \ge 3+\frac{3}{A}$, we have
$$\lambda+\left(\frac{A}{2A+1}(\lambda+1)\right)A \ge \lambda+\frac{1}{3}\lambda A \ge A + \lambda +1.$$
Hence, it is proved that the upper bound of N\"{o}cker's algorithm is larger than our algorithms. 
\end{proof}

\noindent We note that $\lambda$ is usually larger than $\max\left(\frac{4}{3}A+\frac{4}{3},3+\frac{3}{A}\right)$. 
$A$ is usually less than 10 and $\lambda$ is usually more than 100.

\subsection{Optimality Proof for Digit Set $\{\bar{1}, 0, 1\}$}\label{digit_set}

Before analyzing the time model further, we have a proposition that using digit set $\mathcal{S}=\mathcal{C}=\{\bar{1},0,1\}$ is sufficient.

\begin{proposition}
If $0 < D \le A$, binary representation of $n$ using digit set $\mathcal{S} = \mathcal{C}=\{\bar{1},0,1\}$ has no larger $T(N_\mathcal{S},\lambda)$ than all other $\mathcal{S}' \supseteq \{\bar{1},0,1\}$.
\end{proposition}

\begin{proof}
We prove this 
by contradiction. Suppose there is 
$N_\mathcal{S} = n_{\lambda}...n_0$ with some $|n_i|>1$ that has smallest $T(N_\mathcal{S},\lambda)$. Define $t:=T(N_\mathcal{S},i-1)$ (define $t:=0$ if $i=0$) and consider $n_{i+1}n_i$.

\textit{Case $n_i \equiv 1$ (mod 2):} We can construct new 
$N'_\mathcal{S}=n'_\lambda...n'_0$ with $n'_k=n_k$ for all $0 \le k \le \lambda$ except $n'_i=\textnormal{sgn}(n_i)$ and $ n'_{i+1} = n_{i+1}+\frac{n_i-\textnormal{sgn}(n_i)}{2} \le |n_{i+1}|+\frac{|n_i|-1}{2}.$ In the case where $n_i$ is not the least significant non-zero digit of $N_\mathcal{S}$, we have $T(N_\mathcal{S},i) = \max(t,iD)+|n_i|A \ge (i+1)D.$ \linebreak
By that, $T(N_\mathcal{S},i+1) = \max(T(N_\mathcal{S},i),(i+1)D) +|n_{i+1}|A \linebreak = \max(t,iD)+|n_i|A+|n_{i+1}|A.$
By the previous equation and the fact that $T(N'_\mathcal{S},i) = \max(t,iD)+A \ge (i+1)D$, \linebreak 
we have $T(N'_\mathcal{S},i+1) \le \max(T(N'_\mathcal{S},i),(i+1)D)+\linebreak \left(|n_{i+1}|+\frac{|n_i|-1}{2}\right)A \le \max(t,iD)+A+\left(|n_{i+1}|+\frac{|n_i|-1}{2}\right)A \linebreak \le T(N_\mathcal{S},i+1).$
We can use the similar argument to prove that $T(N'_\mathcal{S},i+1) \le T(N_\mathcal{S},i+1)$ for the case where $n_i$ is the least significant non-zero digit of $N_\mathcal{S}$. Because $n_k' = n_k$ for all $k > i + 1$, we have $T(N'_\mathcal{S},\lambda) \le T(N_\mathcal{S},\lambda)$.

\textit{Case $n_i \equiv 0$ (mod 2):} We can construct new 
$N'_\mathcal{S}=n'_\lambda...n'_0$ with $n'_k=n_k$ for all $0 \le k \le \lambda$ except $n'_i=0$ and $ n'_{i+1} = n_{i+1}+\frac{n_i}{2} \le |n_{i+1}|+\frac{|n_i|}{2}.$ Then, we can use the similar argument as in the case when $n_i \equiv 1$ (mod $2$) to show that $T(N'_\mathcal{S},\lambda) \le T(N_\mathcal{S},\lambda)$.

We can repeat changing $n_i$ where $|n_i|>1$ using the above method to get $N'_\mathcal{S}$ using only $\{\bar{1},0,1\}$ with no more $T(N_\mathcal{S},\lambda)$. 
This means that using $\{\bar{1},0,1\}$ is sufficient.
\end{proof}

From Proposition 1, we can assume 
that 
the representation we will consider for optimal representation is $N_\mathcal{C}$. 

\subsection{Delay and Optimal Representation}\label{delay}

We introduce a concept of delay when comparing time of 
processors (finished time of addition processor minus finished time of doubling processor at the same step) as follows:

\begin{definition}[Delay of addition processor in parallel scalar point multiplication]
Let $D=1$ be the time processor used for one doubling, $A\ge D$ be the time used for one addition, and $n \in \mathbb{Z}_+$ with 
$N_\mathcal{C} = n_\lambda...n_0$. The delay of addition processor after calculating $nP$ using $N_\mathcal{C}$ from $n_0$ up to $n_i$ is $\delta(N_\mathcal{C},i)$ which can be calculated by:\\\\
$\delta(N_\mathcal{C},i):=T(N_\mathcal{C},i)-iD=T(N_\mathcal{C},i)-i=$
$$
\begin{cases}
0 & \textrm{if } i=0;\\
\delta(N_\mathcal{C},i-1)-1 & \textrm{if } i>0 \textrm{ and } n_i=0;\\
0 & \textrm{if } i>0 \textrm{, } n_i \neq 0 \textrm{, and }\\
& n_j=0; \forall j, 0 \le j < i;\\
\max(\delta(N_\mathcal{C},i-1)+(A-1),A) & \textrm{otherwise.}
\end{cases}
$$
\end{definition}

The delay after calculating $nP$ using $N_\mathcal{C}$ is $\delta(N_\mathcal{C},\lambda)$. To calculate the delay, we consider each digit from $n_0$ to $n_{\lambda}$. The delay at $n_0$ is 0. When $n_i=0$, only doubling processor does its work, so the delay decreases by $D=1$. When we consider the least significant non-zero digit, the delay is 0 as we copy a result. 
And, when we consider other non-zero digits $n_i$, both processors do their works, so the delay increases by $A-D=A-1$. But, if 
$\delta(N_\mathcal{C},i-1)<D$, doubling processor will finish calculating $2^iP$ after addition processor finishes calculating up to $n_{i-1}$. This means addition processor needs to wait for $2^iP$, and after the addition, the delay is 
$A$.

\begin{example}
Let 
$D=1$, $A=3$, and $n = 29$. If we calculate $29P$ using $N_\mathcal{B} = 011101$ 
(`0' is added to compare with $N_\mathcal{C}$), we have $\delta(N_\mathcal{B},\lambda)=6$. If we use $N_\mathcal{C}=1000\bar{1}\bar{1}$, we have $\delta(N_\mathcal{C},\lambda)=3$. In Section~\ref{proof_algo1}, we will show that this is optimal among all representations.\hfill $\square$
\end{example}

From Definition 3, 
$\delta(N_\mathcal{C},\lambda)+\lambda=T(N_\mathcal{C},\lambda).$ Hence, 
$N^*_\mathcal{C}$ has smallest $T(N_\mathcal{C},\lambda)$ among all $N_\mathcal{C} \in \mathscr{N}_\mathcal{C}$ if and only if it has smallest $\delta(N_\mathcal{C},\lambda)$ among all $N_\mathcal{C} \in \mathscr{N}_\mathcal{C}$ (
using 
same $\lambda$).

\begin{definition}[Optimal canonical binary representation of $n$ for parallel scalar point multiplication]
$N^*_\mathcal{C}$ is an optimal canonical binary representation of $n$ for parallel scalar point multiplication with addition time $A$ if for all $N_\mathcal{C} \in \mathscr{N}_\mathcal{C}$,
$\delta(N^*_\mathcal{C},\lambda) \le \delta(N_\mathcal{C},\lambda).$
We use $\mathscr{N}^*_\mathcal{C}$ to denote set of all $N^*_\mathcal{C}$.
\end{definition}

\section{Optimal Representation when $0 < 2D \le A$}\label{optimal1}

\subsection{Algorithm}\label{algo1}

In the case where $A \ge 2$ ($D=1$), we can construct \linebreak $N^*_\mathcal{C} \in \mathscr{N}^*_\mathcal{C}$ from $N_\mathcal{B}$ using Algorithm 2.

\begin{algorithm}[ht]
\DontPrintSemicolon
\caption{Changing binary representation to optimal representation when $0 < 2D \le A$}
\SetKwInOut{Input}{input}
\SetKwInOut{Output}{output}
\Input{$N_\mathcal{B} = n_{\lambda} ...n_0$}
\Output{$N^*_\mathcal{C} = n'_{\lambda+1}...n'_0 \in \mathscr{N}^*_\mathcal{C}$}
\Begin{
	$\ell \leftarrow$ index of the least significant `1' of $N_\mathcal{B}$\;
    \If {$N_\mathcal{B}$ ends with $11(01)^*010^*$} {
    	$n_\ell \leftarrow \bar{1}$, $n_{\ell+1} \leftarrow 1$\;
    	\Return toNAF($N_\mathcal{B},\ell+1$)\;
    }
    \ElseIf {$N_\mathcal{B}$ ends with $0(01)^*0110^*$} {
    	\Return toNAF($N_\mathcal{B},\ell+1$)
    }
    \lElse {\Return toNAF($N_\mathcal{B},\ell$)}
}
\end{algorithm}


\begin{example}
Consider $n=29$ with $N_\mathcal{B} = 11101$. From Algorithm 2, $\ell=0$ and $N_\mathcal{B}$ ends with $11(01)^*010^*$. We change $N_\mathcal{B}$ to $1111\bar{1}$ and transform to NAF not considering $\bar{1}$ using Algorithm 1. We have $N^*_\mathcal{C} = 1000\bar{1}\bar{1}$ which has smallest scalar point multiplication time for any $A \ge 2D$. \hfill $\square$
\end{example}

We can see that Algorithm 2 has $O(\log_2 n)$ complexity and uses $O(1)$ additional space. Note that Algorithm 2 generates optimal representations not depends on the value $A$.

\subsection{Optimality Proof for Algorithm 2}\label{proof_algo1}


\begin{lemma}
Let $D=1$, $A\ge 1$, $N_\mathcal{C}=n_\lambda...n_0$, and $N'_\mathcal{C}=n'_\lambda...n'_0$. If $n_k=n'_k$ for all $0<i \le k \le j$ with some $n_p \neq 0$, $n'_q \neq 0$ for some $0 \le p,q < i$, and $\delta(N_\mathcal{C},i-1) \ge \delta(N'_\mathcal{C},i-1)$, then $\delta(N_\mathcal{C},j) \ge \delta(N'_\mathcal{C},j)$.
\end{lemma}

\begin{proof}
We prove this lemma by induction on $k$ from $i$ to $j$. Assume that $n_k=n'_k$ and $\delta(N_\mathcal{C},k-1) \ge \delta(N'_\mathcal{C},k-1)$. If $n_k=n'_k=0$, then $
\delta(N'_\mathcal{C},k) = \delta(N'_\mathcal{C},k-1)-1 \le \delta(N_\mathcal{C},k-1)\linebreak -1 = \delta(N_\mathcal{C},k).$ If $n_k=n'_k=1$ or $\bar{1}$, because both are not the least significant non-zero digits, then $\delta(N'_\mathcal{C},k) = \max(\delta(N'_\mathcal{C},k-1)+(A-1),A) \le \max(\delta(N_\mathcal{C},k-1)+(A-1)\linebreak ,A) = \delta(N_\mathcal{C},k).$
\end{proof}

\begin{lemma}
Let $D=1$, $A \ge 2$, $k \ge 1$, $N_\mathcal{C}=n_\lambda...n_0$, and $N'_\mathcal{C}=n'_\lambda...n'_0$. If $n_{i+k}...n_i=01^k$, $n'_{i+k}...n'_i=10^k$, and $\delta(N_\mathcal{C},i-1) \ge \delta(N'_\mathcal{C},i-1) \ge 2$, then $\delta(N_\mathcal{C},i+k) \ge \delta(N'_\mathcal{C},i+k) \ge 2.$
\end{lemma}

\begin{proof}
Define $d_{i-1}:=\delta(N_\mathcal{C},i-1), d'_{i-1}:=\delta(N'_\mathcal{C},i-1)$, $ d_{i+k}:=\delta(N_\mathcal{C},i+k)$, and $d'_{i+k}:=\delta(N'_\mathcal{C},i+k)$. Because $d_{i-1}\ge d'_{i-1} \ge 2$, we know that $n_i$ and $n'_{i+k}$ are not the least significant non-zero digits. From Definition 3, we have $d_{i+k} = d_{i-1}+k(A-1)-1$ and $d'_{i+k} = \max(d'_{i-1}-k+(A-1),A).$

\textit{Case $d'_{i-1}-k < 1:$} Because $d_{i-1} \ge 2$, we have $d_{i-1}-1 \ge 1$ and $2 \le d'_{i+k} = A \le (d_{i-1}-1)+(A-1) \le d_{i-1}-1+k(A-1) = d_{i+k}.$

\textit{Case $d'_{i-1}-k \ge 1:$} We have $2 \le d'_{i+k} = d'_{i-1}-k+(A-1) \linebreak \le d_{i-1}-k+(A-1) \le d_{i-1}-1+k(A-1) = d_{i+k}. $
\end{proof}

\begin{lemma}
If $A \ge 2$, considering from the second least significant non-zero digit, NAF representation has smallest delay among all canonical binary representations.
\end{lemma}

\begin{proof}
We prove this lemma by contradiction. Suppose there is canonical binary representation $N_\mathcal{C} = n_\lambda...n_0$ with consecutive non-zero digits (not consider the least significant non-zero digit) that has smallest delay. Consider the consecutive non-zero digits in four following cases with $k \ge 2$.

\textit{Case $n_{i+k}...n_i=01^k:$} Consider $N'_\mathcal{C}$ with $n'_j=n_j$ for all $0\le j \le \lambda$ except $n'_{i+k}...n'_i=10^{k-1}\bar{1}$. Because $\delta(N_\mathcal{C},i-1)=\delta(N'_\mathcal{C},i-1)$, then $\delta(N_\mathcal{C},i)=\delta(N'_\mathcal{C},i) \ge A \ge 2$ and by Lemma 2, we can conclude that $\delta(N_\mathcal{C},i+k) \ge \delta(N'_\mathcal{C},i+k) \ge 2$. Since $n'_\lambda...n'_{i+k+1}=n_\lambda...n_{i+k+1}$, by Lemma 1, we get $\delta(N_\mathcal{C},\lambda) \ge \delta(N'_\mathcal{C},\lambda)$. 

\textit{Case $n_{i+k}...n_i=0\bar{1}^k:$} The proof is similar to case $01^k$.

\textit{Case $n_{i+1}n_i=1\bar{1}:$} Consider $N'_\mathcal{C}$ with $n'_j=n_j$ for all $0\le j \le \lambda$ except $n'_{i+1}n'_i=01$. Because $\delta(N_\mathcal{C},i-1)=\delta(N'_\mathcal{C},i-1)$, then $\delta(N_\mathcal{C},i+1) =\max(\delta(N_\mathcal{C},i-1)+(A-1),A)+(A-1)$ and $\delta(N'_\mathcal{C},i+1) =\max(\delta(N'_\mathcal{C},i-1)+(A-1),A)-1 \le \delta(N_\mathcal{C},i+1).$ Since $n'_\lambda...n'_{i+2}=n_\lambda...n_{i+2}$, by Lemma 1, we get $\delta(N_\mathcal{C},\lambda) \ge \delta(N'_\mathcal{C},\lambda)$. 

\textit{Case $n_{i+1}n_{i}=\bar{1}1:$} The proof is similar to case $1\bar{1}$.
\end{proof}

Lemmas 1-3 show that NAF 
has the smallest delay 
when considering from the second least significant non-zero digit. However, we can choose where that 
non-zero digit will be from two options: 
$010^*$ ending 
could be changed to $1\bar{1}0^*$, and 
$01^p110^*$ ending ($p \ge 0$) 
could be changed to $10^p0\bar{1}0^*$. This change 
may decrease the delay. We prove this in Theorem 4.

\begin{theorem}[Optimal representation when $0 < 2D \le A$]
Algorithm 2, using the following rules, produces 
$N^*_\mathcal{C} \in \mathscr{N}^*_\mathcal{C}$ for $0 < 2D \le A$.
\begin{itemize}
\item If $N_\mathcal{B}$ ends with $11(01)^*010^*$, change ending $01$ to $1\bar{1}$ and change this to NAF starting at \textnormal{`}$1$\textnormal{'} in this $1\bar{1}$ (the second least significant non-zero digit's position is changed).
\item If $N_\mathcal{B}$ ends with $0(01)^*0110^*$, change this to NAF starting at the second least significant \textnormal{`}$1$\textnormal{'} (the second least significant non-zero digit's position is not changed).
\item Otherwise, change $N_\mathcal{B}$ 
to NAF starting at the least significant \textnormal{`}$1$\textnormal{'} (the second least significant non-zero digit's position is changed if the ending is $11(01)^*0110^*$ or $1110^*$, and is not changed if the ending is $0(01)^*010^*$).
\end{itemize}
\end{theorem}

\begin{proof}
Let $\mathcal{N}$ and $\mathcal{N}'$ be some consecutive digits in canonical binary representation and the least significant `1' of $N_\mathcal{B}$ is at index $\ell$, we will consider each case as follows:

\textit{Case $11(01)^*010^*$ ending:} Let $N_\mathcal{B} = \mathcal{N}11(01)^p010^*$ for some $p \ge 0$. Following Theorem 4, we change $N_\mathcal{B}$ to $\mathcal{N}11(01)^p1\bar{1}0^*$, and after changing to NAF, we have $N^*_\mathcal{C} = \mathcal{N}'00(\bar{1}0)^p\bar{1}\bar{1}0^*$. Consider the case where we do not change the position, after changing $N_\mathcal{B}$ to NAF, we have $N'_\mathcal{C} = \mathcal{N}'0\bar{1}(01)^p010^*$. Hence, $\delta(N'_\mathcal{C},\ell+2p+3) = (A-1)+p(A-2)$ and $\delta(N^*_\mathcal{C},\ell+2p+3) = A+p(A-2)-2 \le \delta(N'_\mathcal{C},\ell+2p+3)$.
Because both prefixes are $\mathcal{N}'$, by Lemma 1, we can conclude that Theorem 4 gives representation $N^*_\mathcal{C}$ with smaller delay. 

For other 4 cases, similar arguments can be applied.
\end{proof}

From Theorem 4, since our $N^*_\mathcal{C}$ is in NAF (except the least significant non-zero digit), we have an upper bound of optimal computation time 
when $0<2D\le A$ as follows.

\begin{corollary}
Let $D=1$, $A \ge 2$, and $N_\mathcal{B}=n_\lambda...n_0$. The upper bound of the parallel scalar point multiplication time using $N^*_\mathcal{C}=n'_{\lambda+1}...n'_0$ from Algorithm 2 is
$T(N^*_\mathcal{C},\lambda+1) \le \left( \frac{1}{2}\lambda+1 \right) A+D.$
\end{corollary}

\begin{proof}
Since our $N^*_\mathcal{C}$ is in NAF except the least significant non-zero digit, the worst representation could be in the form $(10)^{\lambda/2}11$ which has $T(N^*_\mathcal{C},\lambda+1)=\left(\frac{1}{2}\lambda+1 \right) A+D$.
\end{proof}

\noindent Note that $T(N_\mathcal{B},\lambda) \le \lambda A+D$. This bound is tight since it can be achieved from $1^{\lambda+1}$. This means Algorithm 2 generates representations with lower upper bound of computation time.

Moreover, we can see that 
the delay of NAF is different from the optimal no more than 1 (consider case $11(01)^*010^*$ and $0(01)^*0110^*$). Hence, NAF is almost optimal in this case.

\section{Optimal Representation when $0 < D \le A <2D$}\label{optimal2}

\subsection{Algorithm}\label{algo2}

When $1\le A<2$ ($D=1$), we cannot use Algorithm 2 
because Lemmas 2 and 3 do not hold. 
Consider Example 7. 

\begin{example}
Let 
$D=1$, $A=1.2$, and $n = 29$. If we calculate $29P$ using $N_\mathcal{B} = 011101$, we have $\delta(N_\mathcal{B},\lambda)=0.6$. If we use $N_\mathcal{C}=1000\bar{1}\bar{1}$, 
we have $\delta(N_\mathcal{C},\lambda)=1.2$ which is not optimal. 
Using $N_\mathcal{C}=100\bar{1}01$ also gives $\delta(N_\mathcal{C},\lambda)=1.2$. 
\hfill $\square$
\end{example}

Fortunately, we can construct $N^*_\mathcal{C} \in \mathscr{N}^*_\mathcal{C}$ from $N_\mathcal{B}$ in the case where $1\le A<2$ using Algorithm 3.

\begin{algorithm}[ht]
\DontPrintSemicolon
\caption{Changing binary representation to optimal representation when $0 < D \le A < 2D$}
\SetKwInOut{Input}{input}
\SetKwInOut{Output}{output}
\Input{$N_\mathcal{B} = n_\lambda ...n_0$}
\Output{$N^*_\mathcal{C} = n'_{\lambda+1} ...n'_0 \in \mathscr{N}^*_\mathcal{C}$}
\Begin{
	$N^*_\mathcal{C} \leftarrow N_\mathcal{B}$\;
	$\ell \leftarrow$ index of the least significant `1' of $N^*_\mathcal{C}$\;
    $n'_{\lambda+1} \leftarrow 0$\;
    $d \leftarrow 0$\;
    \For {$i \leftarrow \ell+1$ \KwTo $\lambda+1$} {
    	\lIf {$n'_i=1$} {$d \leftarrow \max(d+(A-1),A)$}
        \Else {
        	$d \leftarrow d-1$\;
            \If {$d>A$} {
            	\tcp{"flipping" $n'_i...n'_\ell$ }
            	$n'_\ell \leftarrow \bar{1}$\;
            	\lFor {$j \leftarrow \ell+1$ \KwTo $i-1$} {
                	$n'_j \leftarrow n'_j-1$
                }
                $n'_i \leftarrow 1$\;
            	$d \leftarrow A$\;
                $\ell \leftarrow i$\;
            }
            \lElseIf {$d \le 1$} {$\ell \leftarrow i+1$}
        }
    }
    \Return $N^*_\mathcal{C}$\;
}
\end{algorithm}

In Algorithm 3, we consider each digit from the least significant `1' to $n'_{\lambda+1}$ and calculate the delay at each digit. 
$\ell$ 
keeps the index of the least significant digit still in consideration. When we encounter $n_i=0$, if delay $d> A$, we flip $n'_i...n'_\ell$ from $01...11$ to $10...0\bar{1}$, set $d \leftarrow A$, and set $\ell \leftarrow i$ (start new considering sequence at $n_i$). If delay $d \le 1$, we set $\ell \leftarrow i+1$ (start new sequence at $n_{i+1}$). In the case where $1<d\le A$, we keep going until one of the previous cases occurs. 

\begin{example}
Let $D=1$, $A=1.7$, and $n = 13911$. $N_\mathcal{B} = 011011001010111$ with $n_{\lambda+1}=0$. From Algorithm 3, 
we have
\begin{align*}
i &= 3; & d &= 1.4 \in (1,A] & \ell &= 0 \\
i &= 5; & d &= 1.1 \in (1,A] & \ell &= 0 \\
i &= 7; & d &= 0.8 \le 1 & \ell &= 8 \\
i &= 8; & d &= -0.2 \le 1 & \ell &= 9 \\
i &= 11; & d &= 1.4 \in (1,A] & \ell &= 9 \\
i &= 14; & d &= 1.8 > A.
\end{align*}
We have $d = 1.8 > A$, so we flip $n_{14}...n_9$ from $011011$ to $100\bar{1}0\bar{1}$ and get $N^*_\mathcal{C} = 100\bar{1}0\bar{1}001010111$. We also set $\ell = 14$ and then algorithm terminates. \hfill $\square$
\end{example}

We can see that Algorithm 3 has $O(\log_2 n)$ complexity and uses $O(1)$ additional space.

\subsection{Optimality Proof for Algorithm 3}\label{proof_algo2}


\begin{lemma}
Let $D=1$, $1\le A<2$, $N_\mathcal{B}=n_\lambda...n_0$ and delay $d$ at $\ell-1$ is no more than $1$. Following Algorithm 3, if $n_i=0$, $d>A$, and we flip $n_i...n_\ell$, we have delay $d$ at $i$ equals to $A$.
\end{lemma}

\begin{proof}
We 
define $d^{(0)} \le 1$ as the delay at $\ell-1$. We 
see that 
$n_i...n_\ell$ must be in the form $01^{p_k}...01^{p_2}01^{p_1}$ for some $k>0$ and $p_j>0$ for all $1 \le j \le k$. The sequence has no consecutive zeros because if $1<d \le A$ at the first zero, $d$ is then less than 1 after the second zero and `00' is not in the sequence.

We prove this lemma by induction on $j$ from $1$ to $k$. We define $d^{(j)}$ as the delay after considering up to $01^{p_j}$. 
Consider $01^{p_1}$ when the rightmost 
`1' 
is not the least significant `1' of $N_\mathcal{B}$. If $p_1=1$, 
we have $d^{(1)}\le 1$, which means this case cannot happen. For $p_1>1$, changing $01^{p_1}$ to $10^{p_1-1}\bar{1}$ makes $d^{(1)}=\max(A-(p_1-1)+(A-1),A)=A$ since $A-(p_1-1)<1$. 
When the rightmost 
`1' 
is also the least significant `1' of $N_\mathcal{B}$, we also have $d^{(1)}=A$ for all $p_1>0$.

For $01^{p_j}$, $j>1$, by induction, 
flipping $01^{p_{j-1}}...01^{p_1}$ gives $10^{p_{j-1}}...\bar{1}0^{p_1-1}\bar{1}$ with $d^{(j-1)}=A$. When we flip $01^{p_j}1$ to $10^{p_j}\bar{1}$ with `1' from $10^{p_{j-1}}$, since ${d^{(j-1)}-p_j=A-p_j<1}$, we also have
$d^{(j)} =\max(d^{(j-1)}-p_j+(A-1),A) = A.$
\end{proof}

\begin{lemma}
Let $D=1$, $1\le A<2$, and $N_\mathcal{B}=n_\lambda...n_0$. From 
Algorithm 3, all 
$n_i...n_\ell$ have delay $d$ at $\ell-1$ no more than 1.
\end{lemma}

\begin{proof}
We prove this lemma by induction on each sequence $n_{i_j}...n_{\ell_j}$ from the least to the most significant digit. We first consider the least significant sequence $n_{i_0}...n_{\ell_0}$. Since $n_{\ell_0}$ is the least significant non-zero digit of $N_\mathcal{B}$, the delay $d$ at $\ell_0-1$ is definitely no more than 1. The base case is proved.

By induction, we assume that the considering sequence $n_{i_j}...n_{\ell_j}$ has delay $d$ at $\ell_j-1$ no more than 1. If we do not flip this sequence, we have $d$ at $i_j$ no more than $1$. Because $n_k=0$ for all $i_j+1 \le k \le \ell_{j+1}-1$ (since the next sequence starts at $n_{\ell_{j+1}}$), we have $d$ at $\ell_{j+1}-1$ no more than 1.

If we flip $n_{i_j}...n_{\ell_j}$, by Lemma 4, we have $d$ at $n_{i_j}$ equals to $A$ and from Algorithm 3, the next sequence starts at $n_{i_j}$ (we have $\ell_{j+1}=i_j$). Since $n_{i_j}=1$ and $n_{i_j-1}=0$ with $d$ at $i_j$ equals to $A$, by Definition 3, we have that $d$ at $i_j-1=\ell_{j+1}-1$ must be no more than 1. The induction step is completed.
\end{proof}

\begin{lemma}
Let $D=1$, $1\le A<2$, and $N_\mathcal{B}=n_\lambda...n_0$. Following Algorithm 3, if $n_i=0$ and $d>A$, there is no representation of $n_i...n_\ell$ which gives $d$ at $i$ smaller than $A$.
\end{lemma}

\begin{proof}
We prove this lemma by contradiction. We consider 
$n_i...n_\ell$ in the form $01^{p_k}...01^{p_2}01^{p_1}$ for some $k>0$ and $p_j>0$ for all $1 \le j \le k$. To have $d^{(k)}<A$, `0' in $01^{p_k}$ must be left as `0' because if it is changed to `1', we have $d^{(k)} \ge A$. So, we consider `0' in $01^{p_{k-1}}$ and it must be left as `0' because if it is changed to `1', we have $d^{(k-1)} \ge A$ which makes $d^{(k)} \ge A$ (since before changing, we have $1<d^{(k-1)}\le A$ but $d^{(k)}> A$ already). Using the same reason, we have that `0' in $01^{p_1}$ must not be changed to have $d^{(k)}<A$. This contradicts the assumption since $01^{p_k}...01^{p_2}01^{p_1}$ has $d^{(k)}>A$.
\end{proof}

\begin{theorem}[Optimal representation when $0<D\le A<2D$]
Algorithm 3 produces 
$N^*_\mathcal{C} \in \mathscr{N}^*_\mathcal{C}$ for $0<D\le A<2D$. That is, 
flipping $n_i...n_\ell$ when $d>A$ $(D=1)$ gives 
smallest delay.
\end{theorem}

\begin{proof}
We consider 
$d$ after 
$n_i...n_\ell$ in two cases. If $d \le 1$, flipping 
will give 
more delay, 
so we 
do not flip it. This representation is optimal since there is no other representation with smaller delay than $d$. Otherwise, we have $d>A$, and by Lemmas 4-6, flipping 
will make $d=A$ which is the smallest delay we can achieve. The sequence is always classified in one of these two cases since there are always two consecutive zeros in front of the representation which make $d \le 1$.
\end{proof}

From Theorem 5, we have an upper bound of optimal computation time 
when $0<D\le A<2D$ as follows.

\begin{corollary}
Let $D=1$, $1\le A<2$, and $N_\mathcal{B}=n_\lambda...n_0$. The upper bounds of the delay and parallel scalar point multiplication time using $N^*_\mathcal{C}=n'_{\lambda+1}...n'_0$ from Algorithm 3 are $\delta(N^*_\mathcal{C},\lambda+1) \le A$ and $T(N^*_\mathcal{C},\lambda+1) \le A + \lambda +1$.
\end{corollary}

\begin{proof}
If $n'_{\lambda+1}=0$, it is obvious that $\delta(N^*_\mathcal{C},\lambda+1) \le A$. If $n'_{\lambda+1}=1$, this happens from flipping and we have $\delta(N^*_\mathcal{C},\lambda+1)=A$. Because the delay is no more than $A$ at $n'_{\lambda+1}$, hence $T(N^*_\mathcal{C},\lambda+1) = \delta(N^*_\mathcal{C},\lambda+1)+\lambda+1 \le A+\lambda+1$.
\end{proof}

\noindent Note that $T(N_\mathcal{B},\lambda) \le \lambda A+D$. This bound is tight since it can be achieved from $1^{\lambda+1}$. This means Algorithm 3 generates representations with lower upper bound of computation time.

Moreover, if we change $N_\mathcal{B}$ to NAF, its delay after considering $n_{\lambda+1}$ is also no more than $A$ but may not be optimal (as 
in Example 7). We prove this in the following proposition.

\begin{table*}[ht]
\caption{Computation time of scalar point multiplication and buffer space used with 100,000 random integers from $[1,2^{256}-1]$ }
\centering
{\small
\renewcommand{\arraystretch}{1.2}
\begin{tabular}{|c|c|c|c|c|c|c||c|c|c|c|c|c|} 
 \hline
 \multirow{3}{*}{$\displaystyle\frac{A}{D}$} & \multicolumn{6}{|c||}{Computation Time} & \multicolumn{6}{|c|}{Buffer Space} \\  \cline{2-13} &
 \multicolumn{2}{|c|}{using $N_\mathcal{B}$} & \multicolumn{2}{|c|}{using $N^*_\mathcal{C}$} & \multicolumn{2}{|c||}{using NAF} & \multicolumn{2}{|c|}{using $N_\mathcal{B}$} & \multicolumn{2}{|c|}{using $N^*_\mathcal{C}$} & \multicolumn{2}{|c|}{using NAF} \\ \cline{2-13} 
 & avg. & max & avg. & max & avg. & max & avg. & max & avg. & max & avg. & max \\ \hline
1.00 & \textbf{255.0} & \textbf{256.0} & \textbf{255.0} & \textbf{256.0} & 255.7 & 257.0 & \textbf{1.000} & \textbf{1} & \textbf{1.000} & \textbf{1} & \textbf{1.000} & \textbf{1} \\
1.25 & \textbf{255.5} & 261.0 & \textbf{255.5} & \textbf{257.3} & 255.9 & \textbf{257.3} & 2.682 & 6 & 2.000 & 2 & \textbf{1.000} & \textbf{1} \\
 1.50 & 256.3 & 268.5 & \textbf{255.9} & \textbf{257.5} & 256.2 & \textbf{257.5} & 3.854 & 10 & 2.000 & 2 & \textbf{1.000} & \textbf{1} \\
 1.75 & 258.4 & 291.3 & \textbf{256.3} & \textbf{257.7} & 256.4 & \textbf{257.7} & 5.991 & 22 & 2.000 & 2 & \textbf{1.000} & \textbf{1} \\
 \hline
 2.00 & 268.2 & 326.0 & \textbf{256.7} & \textbf{258.0} & \textbf{256.7} & \textbf{258.0} & 10.238 & 37 & \textbf{1.000} & \textbf{1} & \textbf{1.000} & \textbf{1} \\
 2.25 & 292.2 & 366.5 & \textbf{257.2} & \textbf{264.3} & \textbf{257.2} & \textbf{264.3} & 18.817 & 50 & \textbf{2.043} & \textbf{4} & 2.044 & \textbf{4} \\
 2.50 & 322.1 & 407.0 & \textbf{258.0} & \textbf{274.5} & \textbf{258.0} & \textbf{274.5} & 28.245 & 61 & \textbf{2.742} & \textbf{7} & 2.745 & \textbf{7} \\
 2.75 & 353.3 & 447.5 & \textbf{260.0} & \textbf{289.8} & \textbf{260.0} & \textbf{289.8} & 36.824 & 69 & \textbf{3.974} & \textbf{13} & 3.979 & \textbf{13} \\
 \hline
\end{tabular}}
\end{table*}

\begin{proposition}
Let $N_\mathcal{C}$ be a NAF representation of $n$.
When $D=1$ and $1 \leq A < 2$, we have $\delta(N_\mathcal{C}, i) \leq A$ for all $i \geq 0$. Furthermore, if $n_i = 0$, then $\delta(N_\mathcal{C},i) \leq A - 1$. 
\end{proposition}

\begin{proof}
We prove this proposition by induction on $i$. First, we have 
$\delta(N_\mathcal{C}, 0) = 0$. 
For $i \ge 1$, assume that $\delta(N_\mathcal{C}, i - 1) \leq A$ when $n_{i-1} \neq 0$, and $\delta(N_\mathcal{C}, i - 1) \leq A - 1$ when $n_{i - 1} = 0$.

If $n_i = 0$, then $\delta(N_\mathcal{C}, i) = \delta(N_\mathcal{C}, i - 1) - 1 \le A - 1$.

If $n_i \neq 0$, then, because there is not a consecutive non-zeros in NAF representation, we have $n_{i - 1} = 0$. Hence, $\delta(N_\mathcal{C}, i) = 0$, $\delta(N_\mathcal{C}, i) = A$, or 
$\delta(N_\mathcal{C}, i) = \delta(N_\mathcal{C}, i - 1) + (A - 1)
 \le \delta(N_\mathcal{C}, i - 1) + 1 \leq (A - 1) + 1 = A$.
\qedhere
\end{proof}

Since the optimal delay 
is no more than $A$, the delay from NAF is not larger than the optimal delay 
by more than $A<2$. Hence, NAF is also almost optimal in this case.

\section{Experimental Results}\label{exp}

We compare the parallel computation time and the buffer size required 
when we use 
$N_\mathcal{B}$, 
$N^*_\mathcal{C}$, and NAFs in Table~I. The buffer is used when doubling processor finishes its work before addition processor uses it. 
We keep only what is going to be used, i.e. keep $2^iP$ only when $n_i \neq 0$. The experimental results are obtained from 100,000 random integers from $[1,2^{256}-1]$.

We can see 
that NAFs are not always optimal. 
However, 
the difference from optimal 
is always less than $1\%$. 
Also, the buffer size 
obtained from NAF and our optimal is almost equal. 
Hence, NAF is almost an optimal choice for our model. On the other hand, the results from NAF and $N_\mathcal{C}^*$ are much better than those obtained from $N_\mathcal{B}$, especially when $A / D \ge 2$. The improvement from $N_\mathcal{B}$ in the parallel computation time is as large as $4.4\%$ when $A / D = 2$, and the improvement is as large as $35.9\%$ when $A / D = 2.75$.

\section{Conclusion and Future Work}\label{future}
This paper presents that NAF is almost optimal representation for our proposed time model for ``parallel double-and-add" scalar point multiplication. This is because NAF uses a little more time in the model, nearly the same buffer space, and the same time to generate the representation as our optimal representation. 
Currently, we are aiming to extend our idea to the calculation of the multi-scalar point multiplication, where we want to compute $nP + mQ$ for some integers $n,m$ and elliptic points $P, Q$. Also, we plan to consider the communication time between two processors, and other parallel settings such as 
SIMD or SIMT paradigm \cite{mccool2012structured}.

\section*{Acknowledgment}
This research was mainly done while Kittiphon Phalakarn and Kittiphop Phalakarn did an internship at The University of Tokyo. The authors would like to thank The University of Tokyo, Asst. Prof. Athasit Surarerks and Prof. Hiroshi Imai for facilitating the internship. This work was supported by JST ERATO Grant Number JPMJER1201, Japan. 

\bibliographystyle{ieeetr}
\bibliography{main}

\end{document}